%% file: main_IG2.tex
\documentclass[letterpaper,11pt]{llncs}
\usepackage{amsmath}
\usepackage{amsfonts}
\usepackage{amssymb}
\usepackage{algorithm,multicol}
\usepackage{algpseudocode}
\usepackage[letterpaper,margin=1in]{geometry}
\usepackage{tikz}
\usetikzlibrary{matrix,decorations.pathreplacing,shapes}
\usepackage{mathtools}
\usepackage{footmisc}
\begin{document}

\title{\texttt{Intention Games} \\ \large Towards Strategic Coexistence between Partially Honest and Blind Players}
\author{Aditya Ahuja}
\institute{Department of Computer Science and Engineering \\ 
  Indian Institute of Technology Delhi \\
  \texttt{aditya.ahuja@cse.iitd.ac.in}}

\maketitle

\begin{abstract}
\input{00_abstract}
\end{abstract}

\input{01_introduction}
\input{02_motivation}
\input{03_formaldef}

\input{04_examples}

\input{05_dual-related}

\input{06_conclusion}

\appendix
\input{07a_toymodel}


\newpage

\input{IG-ref.bbl}
\end{document}

%% file: 00_abstract.tex
Strategic interactions between competitive entities are generally considered from the perspective of complete revelation of benefits achieved from those interactions, in the form of public payoff functions and/or beliefs, in the announced games. However, there exist strategic interplays between competitors where the players have a choice to strategise under the availability of \emph{private payoffs}, in similar competitive settings. In this contribution, we propose a formal framework for a competitive ecosystem where each player is permitted to defect from publicly optimal strategies under certain private payoffs greater than announced payoffs, given that these defections have certain acceptable bounds in the long run as agreed by all players. We call this game theoretic construction an Intention Game. We formally define an Intention Game, and notions of participational equilibria that exist in such interactions that permit public defections. We compare Intention Games with conventional strategic form games, and demonstrate a type-theoretic construction of Intention Games. In a partially honest setting, we give Intention Game instances of a Cournot competition, secure interactions between mobile applications, an Internet services' data sourcing competition between Internet service providers through content delivery networks, and a Bitcoin mining competition. We give a use of Intention Games to determine player participation in a cryptographic protocol. Finally, we demonstrate the possibility of a dual model of the Intention Games framework.

%% file: 01_introduction.tex
\section{Introduction}
\label{sec:intro}

Game theory \cite{gt} considers non-cooperative strategic interactions among players with public knowledge of the actions available and payoff structures for the involved parties. Even in the case of games with incomplete information \cite{epis-book}, the framework of the game is consistent with those of complete information games. This framework consists of a direct, disclosed belief \cite{gt-beliefs} on the payoff for each player as a probability distribution on the types of all players, permitting players to strategise under a fair amount of certainty on the underlying game. In every case, all actions are chosen by players individually with selfish interests to maximize their payoff alone. It is this notion of rationality in these non-cooperative settings, that dictates strategic choices and principles of equilibria. \\

\noindent Philosophically, epistemic game theory \cite{epis-book} advocates that players' beliefs of the state of the game, their beliefs about the beliefs of other players, and so on, should be observable. This introduces a fair amount of certainty in the underlying game, in the form of a belief hierarchy. However, it is not yet known how to formally reason in instances of high uncertainty in the true game, when there exists extremely high entropy in the true payoffs of each player, given the players have a choice to reveal only a part of their true payoffs. In such instances, the said belief hierarchy would collapse to a maximum entropy distribution (such as a uniform distribution) on beliefs of any order\footnote{Please see Section \ref{subsec:tig} for details on beliefs under Intention Games and Section \ref{sec:related} for a comparison between epistemic game theory and Intention Games.}, and it is prudent to investigate rational behaviour in such a setting. 

\subsection{The Challenge of Strategic Interactions in Highly Uncertain Games}
\noindent Rational agents my compete amongst themselves with a \emph{partial disclosure of payoffs} achieved out of the competition. There exist computational (and economic) interactions among semi-honest agents where strategic games with partial payoff revelation is a compelling yet unaddressed reality\footnote{We will give concrete examples of partially honest, blind competition in both computational and economic settings, in Section \ref{sec:example}.}. As long as each competing agent is in the knowledge that other participating agents might have payoffs over and above the disclosed payoffs (which is true for the said agent also), it is reasonable to assume applicability of novel game structures in this partially honest setting. Further, game equilibria should dictate that there don't exist players which are being dominantly unfair to other players through excessive defections from the partially disclosed payoffs as agreed upon by all players, which result in losses to the other players. Given this strategic setting, we introduce a new framework to formalize such interactions. \\

\subsection{Our Contribution}
\noindent In order to address the said problem of partially honest competition, we propose a new game theoretic framework called Intention Games. Our framework permits repeated strategic-form games among partially honest players with each player having two categories of payoffs and an \emph{intention} to choose between the two categories. The two categories of payoffs are the publicly declared payoffs, which are less than or equal to private true payoffs, that constitute the second category. Each player also keeps it's intention (choice) to align with either of the two payoffs, private. Moreover, these hidden true payoffs per player might change per iteration of the Intention Game and can potentially result from not just secret, but even spontaneous, random or accidental, contracts/agreements of each player with other players or other hidden parties. We define best responses in this ecosystem permitting public defections and also equilibria of participation as a function of the publicly observed defections in actions for hidden payoff increments. We give four motivating examples, one as an economics' model and three in computational settings, where the Intention Game framework of partially honest behaviour is applicable. We give a use-case of employing Intention Games for participant discovery in a secret sharing protocol. To provide closure to this new game definition, we give a possible dual interpretation of the Intention Game framework. We state that our work is a refinement and extension of \cite{ig-stoc18}. \\

\subsection{Paper Organization}
\noindent This paper is organized as follows. Firstly, we motivate a new game paradigm in a setting of high uncertainty in the true game, in Section \ref{sec:motive}. In Section \ref{sec:def} we formally define an Intention Game, what constitutes a participational equilibrium, when an Intention Game degenerates to a conventional strategic form game, and show how Intention Games can be constructed using types \cite{epis-book} to model payoffs. We then give example settings (one economic model, and three computational scenarios) where Intention Games are an appropriate model for the corresponding strategic interactions, in Section \ref{sec:example}. We also give a use of an Intention Game to define a protocol for determining members of a cryptographic scheme in Section \ref{sec:example}. We give the dual interpretation of Intention Games in Section \ref{sec:selfless}. We compare Intention Games with existing game frameworks in Section \ref{sec:related}. We close the paper in Section \ref{sec:conclude} with concluding arguments and future directions for this work.  

%% file: 02_motivation.tex
\section{Motivating a New Game for Partially Honest and Blind Players}
\label{sec:motive}

In this section, we motivate the need for a new game structure for a strategic interaction between players in situations of high uncertainty in the true game. We introduce players' behaviour, define rationality in settings of high uncertainty of the game payoffs, and elucidate situations where such interactions are meaningful.

\subsection{Introduction to Semi-Honest and Blind Players}

We introduce a new category of players, that possess both of the following characteristics: 
\vspace*{2pt}
\begin{itemize}
\item[$\rightarrowtail$] \textbf{Partially Honest:} Players are allowed to have a secret choice in their true payoff function, between a publicly announced payoff and a higher private payoff.
\item[$\rightarrowtail$] \textbf{Blind:} Players are not privy to the definition of the private payoff of the other players, or choice made by other players for their true payoff function.
\end{itemize}

\noindent To reiterate the above statement via notation, consider that each player $i$ has a choice between two payoffs $u_i(\cdot)$ and $v_i(\cdot)$, where $u_i(\cdot)$ is publicly announced, and $v_i(\cdot)$ is private and higher than or equal to $u_i(\cdot)$ for all action profiles \cite{gt}, given a set of possible action profiles $\mathbf{A}$ \footnote{It is true that $\forall i, \forall \mathbf{a} \in \mathbf{A}, v_i(\mathbf{a}) \geq u_i(\mathbf{a})$.\label{fn:priv-above}}. The \emph{partially honest} characteristic states that each player $i$'s choice between $u_i(\cdot)$ and $v_i(\cdot)$ is made by $i$ and known only to $i$. The second \emph{blind} characteristic states that for each player $j \hspace*{3pt} (\neq i)$, player $i$ knows $u_j(\cdot)$, but neither knows the definition of $v_j(\cdot)$ nor the choice made by $j$ between $u_j(\cdot)$ and $v_j(\cdot)$. \\

\noindent We do note that the stated characteristics for the players are partially overlapping in their definition.

\subsection{Rationality in Semi-Honest Strategic Interactions by Blind Players}

Strategic interactions between players in classical game theory are governed by the principle of \emph{rationality}: each player participates in the game with the objective to maximize its reward from the interaction, which is achieved by maximizing its payoff function. In classical games, this is easy to achieve owing to complete knowledge and/or belief on the underlying game. However, when it comes to interactions where the true game is highly uncertain owing to private choices by players towards their individual payoff functions, blind players can only make rational choices under their individual knowledge on the underlying game. So for each player, the principle of maximising its payoff remains, but the knowledge of the underlying game (in which actions are taken collectively by all players) does not, with a very high probability.

\subsection{Example Situations for Partially Honest Strategic Interactions by Blind Players}
\label{subsec:motive-ex}

\noindent We now give two classes of examples to bring out partially honest and blind strategic interactions. Given the set of possible action profiles $\mathbf{A}$, we will denote a collective action by all players by $\mathbf{a} \hspace*{3pt} (\in \mathbf{A})$. We will denote, for any player $i$, the publicly announced payoff by $u_i(\cdot)$, and the higher private payoff by $v_i(\cdot)$\footref{fn:priv-above}. \\

\noindent As a first example, consider a conventional game $\mathbb{G}$ in a war where the players are individual sovereign nations who are bound by a treaty, say $H_0$, the actions are armed troop deployments in a particular geographic region at a certain time, and payoffs $u_i(\cdot)$ for each country $i$, are functions giving a numerical representation of the victory in the battle. Now given an action profile $\mathbf{a}$ for all countries following $H_0$, the $\mathbf{a}$ can be used as a certificate by each country $i \in H_0$ as a bargaining chip for troop deployment for alternate treaty/treaties $\mathbf{H} \hspace*{3pt} (\neq H_0)$ between $i$ and other countries under $\mathbf{H}$. Note that $\mathbf{H}$ is independent of $H_0$ and consequently the resulting cumulative payoff $v_i(\cdot)$ for each country $i$ is private with respect to $\mathbb{G}$ and above that achieved in $\mathbb{G}$. We will formalize this example as a Cournot \cite{mixed-olig} Intention Game in Section \ref{subsec:icd}. \\

\noindent As a second suite of examples, consider the following settings in a generic computational framework. In a traditional game $\mathbb{G}$, players are computational entities, like mobile applications \cite{app-android}, Internet service providers (ISPs) \cite{cdn-isp-coll}, and Bitcoin miners \cite{btc-selfish}. The public payoff $u_i(\cdot)$ for each player $i$ is the outcome of conformance to a defined protocol: for mobile applications, this involves guaranteeing a service to the consumer; for ISPs this involves serving a client; and for Bitcoin miners this involves mining on the longest unconfirmed part of the blockchain. The private payoff $v_i(\cdot)$ unannounced in $\mathbb{G}$ comes into play once the player $i$ chooses to depart from standard protocol behaviour for an illegitimate increment in reward: for mobile applications, this requires compromising the security of other competitive applications running on the same operating system; for ISPs, this requires redeploying content delivery networks (CDNs) to improve HTTP object download times; and, for Bitcoin miners, this involves engaging in selfish mining on smaller forks in the blockchain for greater cryptocurrency fees. We will see constructions of this set of examples in Sections \ref{subsec:app-game}, \ref{subsec:cdnisp-ig}, and \ref{subsec:btc-game} respectively. \\

\noindent Given motivating examples of partially honest strategic interactions, we would demonstrate that Intention Games can be used for defining protocols for discovering participants in a secret sharing scheme \cite{pvss}, as in Section \ref{subsec:key-game}. \\

\noindent Note that through our examples, we wish to reiterate that Intention Games is a framework of mutual acceptance of dishonest behaviour among involved players. A dishonest strategy for an optimal private payoff by one player might result in suboptimal payoffs for the other honestly participating players. Therefore, this dishonesty must have bounds as the game evolves. We formally capture this notion, for game participation feasibility, though our equilibria in Section \ref{subsec:eqm}. 

%% file: 03_formaldef.tex
\section{The Intention Games Framework}
\label{sec:def}

In this section we give the detailed construction of the Intention Games framework. We cover basic definitions, principles and details of equilibria, degeneration to classical strategic form games, and a type-theoretic construction of Intention Games according to our model, sequentially.

\subsection{Notation and Preliminaries}
\label{subsec:prelim}

Let $\Delta(V)$ denote the set of all probability distributions on universe $V$. We will use the set notation $[z] := \{ 1,2,3, ..., z \}$. For brevity, we will denote the split of a vector $\mathbf{y}$ on an index $i$ as $\mathbf{y} = (y_i, \mathbf{y}_{-i})$ while implicitly preserving the order of elements. The $-i$ denotes all indices except $i$. We will use $\mathbb{R}$ to denote the set of real numbers, $\mathbb{R}_+$ to denote the set of non-negative real numbers, and $\mathbb{Z}_+$ to denote the set of non-negative integers. \\
Note that we will only give outlines of equilibria computation. The complexity arguments are implied from classical game theory, with the Nash Equilibrium \cite{gt} being \texttt{PPAD}-Complete \cite{agt}.

\subsection{Definitions}
\label{subsec:def}

We first give the formal definition of an Intention Game.

\begin{definition}[Intention Game]
\label{def:ig}
An Intention Game $\mathcal{G}$ is a repeated game given by a tuple $\mathcal{G} = ([n], \{A_i\}_{i \in [n]}, \{(u_i,v_i)\}_{i \in [n]}, \{ \xi_i \}_{i \in [n]})$ where \\
\begin{enumerate}
\vspace*{-10pt}
\item $[n]$ is the set of players.
\item $\forall i \in [n], A_i$ is the set of actions available to player $i$. Also $\mathbf{A} := \times_{i \in [n]} A_i$ is the set of all action profiles.
\item Each player $i \in [n]$ has a constant public payoff function $u_i: \mathbf{A} \rightarrow \mathbb{R}$ and a private payoff function which can change per iteration $v_i: \mathbf{A} \rightarrow \mathbb{R}$. \\ 
It is also the case that for all iterations, $\forall i \in [n], \forall \mathbf{a} \in \mathbf{A}, v_i(\mathbf{a}) \geq u_i(\mathbf{a})$.
\item All players collectively agree on a `public image' of the Intention Game: \\ $\mathbf{Im}(\mathcal{G}) = ([n], \{A_i\}_{i \in [n]}, \{u_i\}_{i \in [n]})$.
\item Each player $i \in [n]$ individually considers it's `self reflection' of the Intention Game: $\mathbf{Ref}_i(\mathcal{G}) = ([n], \{A_i\}_{i \in [n]}, \{ v_i \} \cup \{u_j\}_{j \in [n] \setminus \{ i \}})$.
\item Each player $i \in [n]$ strategises according to either the public image of the Intention Game $\mathbf{Im}(\mathcal{G})$ or it's self reflection of the Intention Game $\mathbf{Ref}_i(\mathcal{G})$, under a private `Intention' (choice) $\xi_i \in \{ \mathbf{Im}(\mathcal{G}), \mathbf{Ref}_i(\mathcal{G}) \}$.
\end{enumerate}
\end{definition}
\vspace*{5pt}

\noindent Note the disparity that is intrinsic to an Intention Game: all players \emph{collectively claim to conform to} $\mathbf{Im}(\mathcal{G})$, but \emph{each player} $i \in [n]$ can optimize according to $\mathbf{Ref}_i(\mathcal{G})$. Each player $i \in [n]$, is \textbf{semi-honest} in its private choice between $\mathbf{Im}(\mathcal{G})$ and $\mathbf{Ref}_i(\mathcal{G})$, and \textbf{blind} towards the choices of all other players $j \in [n] \setminus \{ i \}$ between $\mathbf{Im}(\mathcal{G})$ and $\mathbf{Ref}_j(\mathcal{G})$, under an unknown $v_j(\cdot)$. For a toy 2-player example of an Intention Game, please see Appendix \ref{app:toy}. \\

\noindent \emph{A Note on the Spontaneity and Unpredictability of Private Payoffs in each Iteration: } The Intention Games framework permits the private payoff $v_i(\cdot)$ to change arbitrarily (for each player $i$), in every iteration of the game, as long as $\forall i \in [n], \forall \mathbf{a} \in \mathbf{A}, v_i(\mathbf{a}) \geq u_i(\mathbf{a})$. Further, no player $i$ can guess any $v_j(\cdot), j \neq i$ with a non-negligible \cite{crypto} probability, as there are exponential in $|\mathbf{A}|$ possible definitions of $v_j(\cdot)$ (in the worst case). \\

\noindent Now we give the definition of (per player) action profiles which publicly reveal defection from best responses under the public image $\mathbf{Im}(\mathcal{G})$ of the Intention Game $\mathcal{G}$.
\begin{definition}[Defection Partition of Action Profiles]
\label{def:def-partition}
For each player $i \in [n]$, there exists a partition $(\mathbf{A}^+_i, \mathbf{A}^=_i)$ of $\mathbf{A}$ under $\mathcal{G}$ such that \\
$ \mathbf{A}^+_i := \{ \mathbf{a} \in \mathbf{A} : \exists a^d_i \in A_i, v_i(\mathbf{a}) \ge u_i(a^d_i, \mathbf{a}_{-i}) > u_i(\mathbf{a}) \} $ and $ \mathbf{A}^=_i := \mathbf{A} \setminus \mathbf{A}^+_i $. \\
The partition $(\mathbf{A}^+_i, \mathbf{A}^=_i)$ is called the defection partition of action profiles for player $i$.
\end{definition}
\noindent Note that in Definition \ref{def:def-partition} above, we extend the notion of a partition to permit $\mathbf{A}^+_i$ to be empty. Also, intuitively, $\mathbf{A}^+_i$ is the set of action profiles revealing defection by player $i$, with $a^d_i$ being a witness of defection for action profile $\mathbf{a}$. Since $v_i(\cdot)$ can change per iteration of the Intention Game, so can the the corresponding $\mathbf{A}^+_i$. \\

\noindent We now define a bound on how many defecting players are permissible in each iteration of an Intention Game. For simplicity of the definition, we assume that for non-defecting players, the choice of the private payoff is the same as their public payoff, making the private Intention $\xi$ redundant.
\begin{definition}[$k$-Intention Game]
\label{def:k-ig}
A given Intention Game $\mathcal{G}$ is a $k$-Intention Game if in each iteration of the game there exist at most $k$ players $p^+ \subseteq [n], |p^+| \le k$ such that $\forall i \in p^+, \exists \mathbf{a} \in \mathbf{A}, v_i(\mathbf{a}) > u_i(\mathbf{a})$ and for the remaining players $\forall i \in p^= = [n] \setminus p^+$ it is the case that $\forall \mathbf{a} \in \mathbf{A}, v_i(\mathbf{a}) = u_i(\mathbf{a})$.
\end{definition}
\noindent For each iteration of the Intention Game, we call $p^+$ the set of defecting players, and $p^=$ as the set of non-defecting players. It's an easy verification that $\forall i \in p^+, \mathbf{Ref}_i(\mathcal{G}) \neq  \mathbf{Im}(\mathcal{G})$ and $\forall i \in p^=, \mathbf{Ref}_i(\mathcal{G}) = \mathbf{Im}(\mathcal{G})$. \\ 

\noindent For the rest of the paper, we will consider only $1$-Intention Games for the notion of our best response strategies and equilibria. For our examples, we will specify in each case when we have a $1$-Intention Game or a $k$-Intention Game as our context.

\subsection{Best Responses}

\noindent We now give how best responses are defined in the Intention Games ecosystem. Note that these are just reinterpretations and extensions of the underlying principles of the Nash equilibrium.

\begin{definition}[Best Response Set]
\label{def:br-set}
Given a strategy universe $\Omega$, a payoff $f$, and a complementary strategy profile $\mathbf{t}$, the Best Response Set is given by \\ $\textsc{BR}_f(\mathbf{t}) := \{ \omega \in \Omega: \forall \omega' \in \Omega, f(\omega,\mathbf{t}) \geq f(\omega',\mathbf{t}) \}$.
\end{definition}

\begin{definition}[Best Response Profiles]
\label{def:br-prof}
Given the `public image' and each `self reflection' of the Intention Game, the best response profiles are given by \\
$ \mathbb{BR}(\mathbf{Im}(\mathcal{G})) := \{ \mathbf{a} \in \mathbf{A}: \forall i \in [n], a_i \in \textsc{BR}_{u_i}(\mathbf{a}_{-i}) \} $ \\
$ \mathbb{BR}(\mathbf{Ref}_i(\mathcal{G})) := \{ (b_i, \mathbf{a}_{-i}) \in \mathbf{A}:  \mathbf{a} \in \mathbb{BR}(\mathbf{Im}(\mathcal{G})), b_i \in \textsc{BR}_{v_i}(\mathbf{a}_{-i}) \} \hspace*{5pt} \forall i \in [n] $
\end{definition}

\noindent Consider the following intuition of the best responses as given in Definition \ref{def:br-prof}. The best response profiles in the `public image' of the Intention Game are the traditional set of Nash equilibria. However, for each player $i \in [n]$, in it's `self reflection', for best responses, $i$ assumes everyone else is playing the Nash equilibrium, and then $i$ defects by playing best response under the private payoff $v_i(\cdot)$ (to the others' Nash choice). \\

\noindent Next, we give how to catch public defections, via defection partition set membership, given the players are playing best responses.

\begin{theorem}[Best Response Profile Dependencies]
\label{thm:br-depend}
Given an action profile $\mathbf{a} \in \mathbf{A}$, if for some $i$, $\mathbf{Ref}_i(\mathcal{G}) \neq \mathbf{Im}(\mathcal{G})$ and $\mathbf{a} \in \mathbb{BR}(\mathbf{Ref}_i(\mathcal{G}))$, then \\ 
\begin{enumerate}
\vspace*{-10pt}
\item $\mathbf{a} \in \mathbf{A}^+_i \Leftrightarrow \mathbf{a} \notin \mathbb{BR}(\mathbf{Im}(\mathcal{G}))$.
\item $\mathbf{a} \in \mathbf{A}^=_i \Leftrightarrow \mathbf{a} \in \mathbb{BR}(\mathbf{Im}(\mathcal{G}))$.
\end{enumerate}
\end{theorem}
\begin{proof}
Since for each player $i$, the best response choices between $\mathbf{Im}(\mathcal{G})$ and $\mathbf{Ref}_i(\mathcal{G})$ only differ in the payoff of player $i$, our proof will only consider choices as a function of $u_i(\cdot)$ and $v_i(\cdot)$. \\ 
\textsc{Proving $1 \Rightarrow$}. Let's say $\mathbf{a} \in \mathbf{A}^+_i$. Then there exists a witness of defection $a^d_i$ such that $u_i(a^d_i, \mathbf{a}_{-i}) > u_i(a_i, \mathbf{a}_{-i})$. So $a_i$ is not a best response under payoff $u_i(\cdot)$ given the complementary action profile $\mathbf{a}_{-i}$. Thus $a_i \notin \textsc{BR}_{u_i}(\mathbf{a}_{-i})$ and $\mathbf{a} \notin \mathbb{BR}(\mathbf{Im}(\mathcal{G}))$. \\
\textsc{Proving $1 \Leftarrow$}. Let's say $a_i$ is not a member of the best response set under payoff $u_i(\cdot)$ given the complementary action profile $\mathbf{a}_{-i}$. Then there exists an $a^d_i \in \textsc{BR}_{u_i}(\mathbf{a}_{-i})$ such that $u_i(a^d_i, \mathbf{a}_{-i}) > u_i(a_i, \mathbf{a}_{-i})$. Also since player $i$ is playing best responses under payoff $v_i(\cdot)$ given complementary action profile $\mathbf{a}_{-i}$ it is true that $v_i(a_i, \mathbf{a}_{-i}) \ge v_i(a^d_i, \mathbf{a}_{-i}) \ge u_i(a^d_i, \mathbf{a}_{-i})$ (the second part of the inequality is true as by definition, for any player, $v_i(\cdot)$ is always greater than or equal to $u_i(\cdot)$ under the same action profile, in this case $(a^d_i, \mathbf{a}_{-i})$). These inequalities imply that $v_i(a_i, \mathbf{a}_{-i}) \ge u_i(a^d_i, \mathbf{a}_{-i}) > u_i(a_i, \mathbf{a}_{-i})$, and consequently $(a_i, \mathbf{a}_{-i}) = \mathbf{a} \in \mathbf{A}^+_i$. \\
\textsc{Proving $2$}. This statement is the equivalence complement (for propositions $\rho_1$ and $\rho_2$, $\rho_1 \Leftrightarrow \rho_2$ if and only if $\neg{\rho_1} \Leftrightarrow \neg{\rho_2}$) of statement $1$, which has been proved.
\end{proof}

\noindent We also give the implication of the dishonest player's actions on the honest players.

\begin{corollary}[Fallout for Honest Players]
\label{cor:fallout}
Given for some dishonest player $i$, $\mathbf{a} \in \mathbb{BR}(\mathbf{Ref}_i(\mathcal{G}))$ and $\mathbf{a} \in \mathbf{A}^+_i$, then $\mathbf{a}$ is a suboptimal payoff action profile for all honest players $j (\neq i)$, as $j$ plays as per $\mathbf{Im}(\mathcal{G})$, but $\mathbf{a} \notin \mathbb{BR}(\mathbf{Im}(\mathcal{G}))$.
\end{corollary}
\noindent Note that in Corollary \ref{cor:fallout} as $\mathbf{a} \in \mathbf{A}^+_i$, there exists a defection witness $a^d_i \in A_i$  that maximizes $u_i(a^d_i, \mathbf{a}_{-i})$, corresponding to the Nash optimal strategy $(a^d_i, \mathbf{a}_{-i})$ under $\mathbf{Im}(\mathcal{G})$. Further, any player $j (\neq i)$ cannot play best responses under $u_j(\cdot)$ as it (mistakenly) assumes the complementary action profile (see Definition 4) to be $(a^d_i, \mathbf{a}_{-\{i,j\}})$ instead of $(a_i, \mathbf{a}_{-\{i,j\}})$, given that the latter one is the one being played. \\

\noindent Given the framework of an Intention Game, we now define how players can decide on coexistence through appropriate equilibria that reflect the number of defections and degrees of defections (through numerical measures on payoffs), in the next subsection.

\subsection{Equilibria in Intention Games}
\label{subsec:eqm}

We first define a participation equilibrium which is captures how many cumulative instances\footnote{In future, we would like to consider an equilibrium definition capturing defection centrality: is there a subset of players defecting disproportionately as compared to other players?} of publicly observed defections from $\mathbf{Im}(\mathcal{G})$ are seen by all players upto the current run of the Intention Game. 

\begin{definition}[Honesty Equilibrium]
\label{eqm:honesty}
A pure-strategy profile vector $(\dot{\mathbf{s}}^t)_{t \in [\tau]} \in \mathbf{A}^\tau$ is a $(\tau,\delta)$-Honesty Equilibrium if after $\tau$ iterations of the Intention Game, \\
given that $\forall t \in [\tau], \forall i \in [n], \dot{\mathbf{s}}^t \in \mathbb{BR}(\mathbf{Im}(\mathcal{G})) \text{ or } \dot{\mathbf{s}}^t \in \mathbb{BR}(\mathbf{Ref}_i(\mathcal{G})) $, \\
it is the case that  $| \{ \dot{\mathbf{s}}^t : t \in [\tau], \dot{\mathbf{s}}^t \notin \mathbb{BR}(\mathbf{Im}(\mathcal{G})) \} | = \delta$.
\end{definition}

\subsubsection*{Computation.} We assume that the computation of a pure-strategy Nash equilibrium for $\mathbf{Im}(\mathcal{G})$ is a given. We give the method for computing the Honesty equilibrium, as an invariant under $t \in [\tau]$. Let's say $\delta_{t - 1}$ is the Honesty equilibrium bound upto epoch $t-1$. Now given $\dot{\mathbf{s}}^t$, compute public defection, using Theorem \ref{thm:br-depend}, by testing membership of $\dot{\mathbf{s}}^t$ in $\mathbf{A}^+_i$ for each $i$. Note that since each player is playing best responses under it's private payoff, this membership can be tested by only finding an $s^d_i \in A_i$ such that $u_i(s^d_i, \dot{\mathbf{s}}^t_{-i}) > u_i(\dot{s}^t_i, \dot{\mathbf{s}}^t_{-i})$. If there exists a single player for which $\dot{\mathbf{s}}^t$ is (publicly) defecting, set $\delta_t = \delta_{t - 1} + 1$. Otherwise set $\delta_t = \delta_{t - 1}$. \\

\noindent We now give a mixed-strategy participation equilibrium for a $1$-Intention Game where the defecting player $i^+$ persists with an unchanging higher payoff $v_{i^+}(\cdot)$ for polynomially (in $n$) many rounds. For each of those rounds, players $[n] \setminus \{i^+\}$ are non-defecting.

\begin{definition}[Defection Equilibrium]
\label{eqm:defection}
A mixed-strategy profile vector \\ $(\dot{\mathbf{s}}_i)_{i \in [n]} \in \Delta^n(\mathbf{A})$, where $\dot{\mathbf{s}}_i$ is a mixed-strategy best response under $\mathbf{Ref}_i(\mathcal{G})$, is a $\mu$-Defection Equilibrium if $\exists i \in [n],$ 
$\mathbf{E}[v_i(\dot{\mathbf{s}}_i) - u_i(\dot{\mathbf{s}}_i)] \ge \mu$ under $\mathbf{Im}(\mathcal{G})$.
\end{definition}

\subsubsection*{Computation.} We assume that the computation of a mixed-strategy equilibrium $\dot{\mathbf{s}}_i$ for $\mathbf{Ref}_i(\mathcal{G})$ is a given: first we compute the mixed-strategy Nash equilibrium under $\mathbf{Im}(\mathcal{G})$ and then replace the $i$th player's (randomized) Nash optimal with the randomized best response under $v_i(\cdot)$ (by keeping the Nash optimal constant for all $j \neq i$). \\
We first give the method by which a player $i$ can compute his own defection bound $\mu_i$. Given the distribution $\dot{\mathbf{s}}_i$ computed in the previous step, it is straightforward to compute the distribution $v_i(\dot{\mathbf{s}}_i) - u_i(\dot{\mathbf{s}}_i)$, if the functions $u_i(\cdot),v_i(\cdot)$ are deterministic and efficient (polynomial time in $|\mathbf{A}|$ computable). So, we can determine $\mu_i$ as the expected value of $v_i(\dot{\mathbf{s}}_i) - u_i(\dot{\mathbf{s}}_i)$. \\ 
We now give the method whereby a player $i \hspace*{3pt} (\neq j)$ can compute a lower bound for $\mathbf{E}[v_j(\dot{\mathbf{s}}_j) - u_j(\dot{\mathbf{s}}_j)]$ given a sufficiently long stream $\tau$ of realizations of $\dot{\mathbf{s}}_j$. For an arbitrary iteration $t$ of the Intention Game, let $\mathbf{r}^t$ be the realization of the strategies played by all players. If $\mathbf{r}^t \in \mathbf{A}^+_j$, find an $a^{d*}_j \in A_j$ such that $c^t_j := u_j(a^{d*}_j, \mathbf{r}^t_{-j}) - u_j(\mathbf{r}^t)$ is maximized. If $\mathbf{r}^t \in \mathbf{A}^=_j$, $c^t_j := 0$. By the law of large numbers, $\mu_j := \frac{\sum_{t \in [\tau]} c^t_j}{\tau}$. \\ 
Finally, $\mu := \texttt{max}_{k \in [n]} \hspace*{5pt} \mu_k$. 

\subsubsection*{Discussion.} It is clear from the definition of both the Honesty and Defection equilibria, the game is more fair as long as $\delta$ and $\mu$ are small. So these equilibria definitions can be used by each player to announce the terms of competition. For instance, players might agree on a $(\tau,\delta)$-Honesty equilibrium conforming game as long as $\delta \leq \delta_0, \forall \tau$, for some contractual constant $\delta_0$. As another case, players might agree on a $\mu$-Defection equilibrium conforming game as long as $\mu \leq \mu_0$ for some previously announced constant $\mu_0$. Whenever $\delta_0, \mu_0$ are exceeded, players terminate the Intention Game.

\subsection{Comparison with Conventional Strategic Form Games}
\label{subsec:normal-degen}

There is an instance when an Intention Game is identical to a conventional (underlying) strategic form game. Consider the case where for all iterations of the Intention Game $\mathcal{G}$, $\forall \mathbf{a} \in \mathbf{A}, v_i(\mathbf{a}) = u_i(\mathbf{a})$, for all players $i$. In this case, for all iterations of $\mathcal{G}$,
$\forall i \in [n], \mathbf{Ref}_i(\mathcal{G}) = \mathbf{Im}(\mathcal{G})$. Further, the Nash equilibrium will hold per iteration of $\mathcal{G}$ and any evolution of the Intention Game would result in $(\tau, 0)$-Honesty and $0$-Defection equilibria. This can be intuitively seen from the fact that for all players $i \in [n]$, in any evolution of the Intention Game, the set of public defection action profiles $\mathbf{A}^+_i$ will always be empty. \\

\subsection{Uncaught Defection and Equilibrium Match under certain Private Payoffs}
\label{subsec:equi-degen}

There can be instances where the best responses for an Intention Game are identical to those of the underlying public strategic form game, even when the two games are different. We give the function family $\{ (u_i,v_i) : \forall \mathbf{a} \in \mathbf{A}, v_i(\mathbf{a}) = c \times u_i(\mathbf{a}), c \in (1, \infty) \}$ for every player $i \in [n]$ in the Intention Game $\mathcal{G}$. Now it is an easy verification that although $\forall i \in [n], \mathbf{Ref}_i(\mathcal{G}) \neq \mathbf{Im}(\mathcal{G})$ (as $\forall \mathbf{a} \in \mathbf{A}, v_i(\mathbf{a}) > u_i(\mathbf{a})$), we have $\mathbb{BR}(\mathbf{Ref}_i(\mathcal{G})) = \mathbb{BR}(\mathbf{Im}(\mathcal{G}))$. Here, the defecting player $i$ would never be caught, as $\mathbf{A}^+_i$ is empty, and computing the Nash equilibrium for $\mathbf{Im}(\mathcal{G})$ would suffice. Consequently, any evolution of the Intention Game would result in $(\tau, 0)$-Honesty and $0$-Defection equilibria. \\

\subsection{A Type-theoretic Construction of Intention Games}
\label{subsec:tig}

We now give an alternate construction of Intention Games, by introducing types \cite{bgi} to capture the space of possible payoff functions, both public and private, corresponding to each player. Consistent with the type theoretic model in Bayesian games \cite{bgi}, we also introduce an entity `Nature', that signals the type for each player, corresponding to the payoff that Nature wishes for the corresponding player. We assume that each signal by Nature is \emph{superior}: for each player, the payoff corresponding to the signal is higher than or equal to a given public default payoff, as is captured formally in the following definition.

\begin{definition}[Set of Superior Type Vectors]
\label{def:suptypes}
Given, for each player $i \in [n]$, a set of types $\Theta_i$, a default type $\theta^0_i \in \Theta_i$, and a payoff function $w_i: \Theta_i \times \mathbf{A} \rightarrow \mathbb{R}$, under type vector set $\mathbf{\Theta} := \times_{i \in [n]} \Theta_i$ and the default vector $\mathbf{\theta}^0 = (\theta^0_1, \theta^0_2, ..., \theta^0_n)$, the set of superior type vectors is given by $\texttt{Sup}_{\theta^0}(\mathbf{\Theta}) = \{ \mathbf{\theta} := (\theta_1, \theta_2, ..., \theta_n) \in \mathbf{\Theta} : \forall i \in [n], \forall \mathbf{a} \in \mathbf{A}, w_i(\theta_i,\mathbf{a}) \geq w_i(\theta^0_i,\mathbf{a}) \}$.
\end{definition}

\noindent We give the formal definition of a Type-theoretic Intention Game, which is an equivalent construction of the Intention Game given in Definition 1. In the following definition, for every player's payoffs (both public and private), we assume the pre-image and image spaces of the payoff are finite. More specifically, we assume each payoff maps action profiles to a (perhaps large) finite field $\mathbb{F}_q := \{ 0, 1, 2, ..., q-1 \}$. Each player has the same set of types, spanning all possible functions from $\mathbf{A}$ to $\mathbb{F}_q$. Also, we will use $\mathcal{U}(V)$ to denote uniform distribution over universe (set) $V$. \\

\begin{definition}[Type-theoretic Intention Game]
\label{def:ig}
The Type-theoretic Intention Game $\mathcal{TIG}$ is a repeated game given by \\ $\mathcal{TIG} = ([n], \{A_i\}_{i \in [n]}, \{\Theta_i\}_{i \in [n]}, \{w_i\}_{i \in [n]}, \{\theta^0_i\}_{i \in [n]}, \{ \xi_i \}_{i \in [n]})$, where: \\
\begin{enumerate}
\vspace*{-10pt}
\item $[n]$ is the set of players.
\item $\forall i \in [n], A_i$ is the (finite) set of actions available to player $i$. Also $\mathbf{A} := \times_{i \in [n]} A_i$ is the set of all action profiles.
\item $\forall i \in [n], \Theta_i := [q^{|\mathbf{A}|}]$ is the set of types for player $i$. Also $\mathbf{\Theta} := \times_{i \in [n]} \Theta_i$ is the set of all type vectors.
\item Each player $i \in [n]$ has a payoff function $w_i: \Theta_i \times \mathbf{A} \rightarrow \mathbb{F}_q$, where: \\ given $\forall i \in [n], \theta^0_i \in \Theta_i$ all players agree on $\mathbf{Im}(\mathcal{TIG}) := \{ w_i(\theta^0_i, \cdot ) \}_{i \in [n]}$.
\item Nature generates, for each player $i \in [n]$, a private, one-time, random permutation $\eta_i: \Theta_i \rightarrow \Theta_i$ that changes in each iteration. For each player $i \in [n]$, Nature reveals $\eta_i$ (to $i$ alone).
\item Given a superior signal $\mathbf{\theta}^\mathcal{N} = (\theta^\mathcal{N}_1, \theta^\mathcal{N}_2, ..., \theta^\mathcal{N}_n) \in \texttt{Sup}_{\theta^0}(\mathbf{\Theta})$ from Nature, Nature generates  $\tilde \theta = (\tilde \theta_1, \tilde \theta_2, ..., \tilde \theta_n) \in \mathbf{\Theta}$ such that $\forall i, \eta_i(\tilde \theta_i) = \theta^\mathcal{N}_i$, and players see $\tilde \theta$.
\item $\forall i \in [n]$, player $i$ recovers $\theta^\mathcal{N}_i$ (alone) from $\tilde \theta$ using $\eta_i$.
\item $\forall i \in [n]$, the private Intention of player $i$ to choose the game, is given by $\xi_i \in \{ w_i(\theta^\mathcal{N}_i, \cdot ), w_i(\theta^0_i, \cdot ) \}$ while fixing $\{ w_j(\theta^0_j, \cdot ) \}_{j \in [n] \setminus \{i\} }$
\end{enumerate}
\end{definition}

\noindent \textbf{Discussion.} Given a private, random permutation on types $\eta_i$, for each player $i \in [n]$, for any distribution $\mathbf{p}_{\texttt{Sup}_{\theta^0}(\mathbf{\Theta})}$ of choices made by Nature, each player $i$ sees the uniform distribution $\mathcal{U}(\texttt{Sup}_{\theta^0_{-i}}(\mathbf{\Theta}_{-i}))$ on the possible private payoffs of other players. It is easy to see that the worst case of zero default payoffs\footnote{Consider the case $\forall i \in [n], \forall \mathbf{a} \in \mathbf{A}, w_i(\theta^0_i,\mathbf{a}) := 0$.}, the set of superior types for all competitors of $i$ is $\texttt{Sup}_{\theta^0_{-i}}(\mathbf{\Theta}_{-i}) = \mathbf{\Theta}_{-i}$ (with size $q^{(n - 1)|\mathbf{A}|}$). Consequently, in the worst case, player $i$ would consider the uniform distribution on the superior payoffs of its competitors $\mathcal{U}(\{ w_j \}_{j \in [n] \setminus \{ i \}})$, resulting in the uniform distribution on the possible strategies by its competitors: $\mathcal{U}(\mathbf{A}_{-i})$, which is a maximum entropy first-order belief \cite{epis-book}. By a similar reasoning, the second order belief of any player $i$ is the uniform distribution (due to max entropy on Nature's signal towards the competitors) on possible actions of its competitors $-i$, and the possible beliefs (which is again only the uniform distribution by our construction) of competitors: $\mathcal{U}(\mathbf{A}_{-i} \times \mathcal{U}(A_i))$. In the worst case, it can be inductively shown that the belief hierarchy \cite{epis-book} collapses due to max-entropy on beliefs of any order: the \emph{expected distribution} of the belief of any order is uniform, implied by uniform distribution on the superior types for each competitor.
 
\subsection*{Collusions for Higher Payoffs}
Given an arbitrary $k$-Intention Game, since the players are both partially honest and blind, it is a strong possibility that upto $k$ players conspire to define each of their private payoffs, and then defect collectively from the public image of the Intention Game for better payoffs. We will see a concrete realization of this possibility in the Bitcoin mining competition (Section \ref{subsec:btc-game}).

\subsection*{Equality Gap and Social Welfare}
The participation equilibria in Section \ref{subsec:eqm} define within themselves what constitutes acceptable bounds on defections through the measures $\delta$ and $\mu$. We state that there can be more measures introduced to ensure social welfare, such as measures to capture any \emph{centrality} in defections: is there a subset of players who are being disproportionately dishonest as compared to the other players? Also, can there exist metrics as a function of chosen true payoffs (between public and private) to formally define social welfare? We state that such formalizations are beyond the scope of this introductory work on Intention Games.

%% file: 04_examples.tex
\section{Example Constructions and a Use of Intention Games}
\label{sec:example}

We now give two classes of example constructions of Intention Games. The first class contains an economic model of competition that is an extension of a Cournot duopoly \cite{gt}. The second class contains three skeletal examples from different computational settings, including a competition for secure interaction between mobile applications, a competition between ISPs for sourcing data from CDNs, and a strategic interaction between Bitcoin miners having a choice between traditional and selfish mining. We also give a use case of Intention Games for a protocol to determine participation in a secret sharing scheme. \\
In each example/use, for each player $i \in [n]$, given the set of action profiles $\mathbf{A}$, we assume $u_i: \mathbf{A} \rightarrow \mathbb{R}_+$ is the public payoff function, and $v_i: \mathbf{A} \rightarrow \mathbb{R}_+$ is the private payoff function. Also, we suppress the Intention set $\{ \xi_i \}_{i \in [n]}$ in each Intention Game definition and informally argue how a choice between the public image and self reflection of the Intention Game affects the players. \\

\subsection{An Intention Cournot Duopoly}
\label{subsec:icd}

We extend a conventional Cournot duopoly where two firms compete in the supply of a single homogenous commodity to a single market with identical cost functions and symmetric payoffs \cite{mixed-olig}. Our extended `Intention Cournot Duopoly' (ICD) involves a secret contract with a preferential but hidden second market which is in contact with at most one firm at any point in time. Both firms have a symmetric secret contract with this hidden preferential market which defines a (higher than public payoff) private payoff in the Intention Game whenever the corresponding firm participates in the contract (given by its Intention $\xi$) in the event of a contact with the hidden market (see Figure 1). \\

\tikzstyle{Firm}=[circle,draw=blue,fill=blue!20,inner sep=2.5mm]
\tikzstyle{Mar}=[rectangle,draw=red,fill=red!20, text width=7em, inner sep=2.6mm]
\tikzstyle{HMar}=[cloud, draw=orange,cloud puffs=10,cloud puff arc=120, aspect=2, text width=7em, inner ysep=1em]

\begin{figure}
\label{fig:icd}
\centering
\begin{tikzpicture}[scale=0.5][domain=0:8]

\draw (1,4) node[Firm,label=center:$1$,label=left:$v_1$] (f_1) {};
\draw (9,4) node[Firm,label=center:$2$,label=right:$v_2$] (f_2) {};
\draw (5,0) node[Mar,label=center:\text{Market \textbf{M}}, label=below:\text{$(u_1,u_2)$}] (m_1) {};
\draw (5,8) node[HMar,label=center:\text{Hidden Market \textbf{H}}] (m_2) {};

\draw (f_1)  to node [label=left:$q_1$] {} (m_1);
\draw (f_2)  to node [label=right:$q_2$] {} (m_1);
\draw[dotted] (f_1)  to node [label=left:$H_1$] {} (m_2);
\draw[dotted] (f_2)  to node [label=right:$H_2$] {} (m_2);

\end{tikzpicture}
\caption{An Intention Cournot Duopoly between Firms $1$ and $2$, given a secret contract $(H_1,H_2)$ with a Hidden Market.}
\end{figure}

\noindent More formally, given a market \textbf{M} and a hidden preferential market \textbf{H} with a secret contract $(H_1,H_2)$, 
$\mathcal{C} := ([2], \{A_1,A_2\}, \{ (u_1,v_1), (u_2,v_2) \})$ is a Intention Cournot Duopoly where \\
\vspace*{-7pt}
\begin{enumerate}
\item $A_i := \{ q_i: q_i \in [0, 1] \}$, \hspace*{3pt} $\forall i \in [2]$.
\item $u_i(q_1,q_2) := q_iP(q_1,q_2) - C(q_i) := q_i(1 - q_1 - q_2) - \frac{1}{2}q_i^2$, \hspace*{3pt} $\forall i \in [2]$ (payoff from \textbf{M}).
\item $H_i:$ If \textbf{H} is in contact with firm $i$, it will take the same supply $q_i$ as to \textbf{M}, compensate production costs, and pay half of the supply from $i$ to it. Otherwise no trade. \hspace*{5pt} ($\forall i \in [2]$) \\
$ h_i(q_1,q_2) := \frac{1}{2}q_i $ or $0$ under contract $H_i$, $\hspace*{3pt} \forall i \in [2]$. 
\item $v_i(q_1,q_2) := u_i(q_1,q_2) + h_i(q_1,q_2)  \hspace*{3pt} \forall i \in [2]$ \\
\end{enumerate}

\noindent Note that the price $P$ and cost $C$ functions are expanded inline in point $3$. Also $-i$ is a single player $3 - i$. 

\subsubsection{Honesty and Defection Equilibria.}
Using superscripts $u*$ and $v*$ for best responses under public and private payoffs respectively, it is an easy calculation that the best response profile sets are singletons: \\
$\mathbb{BR}(\mathbf{Im}(\mathcal{C})) = (q^{u*}_1 = \frac{1}{4}, q^{u*}_2 = \frac{1}{4})$ and \\
$\mathbb{BR}(\mathbf{Ref}_i(\mathcal{C})) = (q^{v*}_i = \frac{5}{12}, q^{u*}_{-i} = \frac{1}{4})$ whenever $H_i$ is applicable and succeeds. \\

\noindent Note that $q^d_i = \frac{1}{4}$ is a (maximal) defection witness of $(q^{v*}_i = \frac{5}{12}, q^{u*}_{-i} = \frac{1}{4})$. \\

\noindent According to our statement of the ICD, at any iteration of $\mathcal{C}$, $H_i$ can succeed for at most one $i$, so $\mathcal{C}$ is a $1$-Intention Game. Thus, the $\delta$ factor of an Honesty Equilibrium can increment by $1$ in an iteration of the ICD if for the said iteration $H_i$ succeeds and consequently $\mathbf{Im}(\mathcal{C}) \neq \mathbf{Ref}_i(\mathcal{C})$. \\

\noindent Also $\forall i \in [2], \mu_i := u_i(q^d_i, q^{u*}_{-i}) - u_i(q^{v*}_i, q^{u*}_{-i}) = 0.09375 - 0.05208 = 0.04167$. So given best responses under $\mathbf{Im}(\mathcal{C})$ and $\mathbf{Ref}_i(\mathcal{C})$, there exists a $(\mu = 0.04167)$-Defection Equilibrium. \\

\noindent Finally, note that a defection by player $i$ costs player $-i$: $u_{-i}(q^{v*}_i, q^{u*}_{-i}) = 0.05208 < 0.09375 = u_{-i}(q^{u*}_i, q^{u*}_{-i})$. \\

\noindent \textbf{Correspondence between $\mathcal{C}$ and Competition between Countries at War.} We now give a correspondence between an Intention Cournot Duopoly $\mathcal{C}$ and two enemy nations engaged in a war. The players in $\mathcal{C}$ correspond to two nations, the actions in $\mathcal{C}$ represent quantity of troop deployment in hostile territory, and the payoffs in $\mathcal{C}$ represent rewards out of the battle, given a treaty \textbf{M} \footnote{We referred to a treaty \textbf{M} as $H_0$ in Section \ref{subsec:motive-ex}.}. The cost function $C$ defines the cost of troop deployment, and the price function $P$ gives the earnings in battle (such as territory captured), per unit of troop deployment. The hidden market \textbf{H} represents alternate treaties each nation (player) might have with other allied nations (not part of $\mathcal{C}$) for troop support in a nearby conflict zone, close to the deployment zone as in \textbf{M}. Thus we have an Intention Game model of strategic wartime interactions between enemy and allied nations. \\

\subsection{An Intention Game for Secure Interactions between Mobile Applications}
\label{subsec:app-game}
\noindent Mobile applications on the Android operating system use a message passing system based on \emph{intents} to communicate \cite{app-android}. These intents are delivered to application components, which include user interface \emph{activities}, background \emph{services}, and \emph{broadcast receivers} \cite{app-android}. There exist well studied legitimate and illegitimate uses of intents and application components to define application (player) behaviour. We list the possible legitimate and illegitimate uses of intents and application components in Table 1. \\

\begin{table}[ht]
\centering
\caption{Possible Legitimate/Malicious Actions by Android based Mobile Applications \cite{app-android}}
\begin{tabular}{|c||c|c|}
\hline
Game Notation & Action Type & Action Description \\
\hline
\hline
\textsc{userexp} & Optimize User Experience & Legitimate behaviour to serve the user. \\
\textsc{btheft} &Broadcast Theft&Broadcast evesdropping and denial-of-service.\\
\textsc{ahijack} &Activity Hijacking&A malicious activity is launched.\\
\textsc{shijack} &Service Hijacking&A malicious service intercepts intent(s) for a legit service.\\
\textsc{binject}&Broadcast Injection&Malicious receivers are vulnerable to malicious broadcasts.\\
\textsc{b-nocheck}&System Broadcast w/o Action Check&Attacks leveraging unchecked action field in intents.\\
\textsc{alaunch}&Activity Launch&Applications can launch malicious activities using intents. \\
\textsc{slaunch}&Service Launch&Applications can maliciously launch unprotected services. \\
\hline
\end{tabular}
\end{table}

\noindent We construct a $k$-Intention Game $\mathcal{G}_{app}$ to formally capture the legitimate/illegitimate interaction between mobile applications. The actions $A_i$ capture the traditional / malicious actions of the mobile applications. The best possible reward for legitimate behaviour by application $i$ is given by $c_{\tau,i}$. The best possible reward for malicious behaviour by application $i$ is given by $\tilde c_{\tau,i}$. Both these rewards can be dependant on the state of the Android operating system in iteration $\tau$ of the Intention Game, given the security measures deployed by the operating system in the said iteration. In the private payoffs for malicious behaviour, the coefficient for $c_{\tau,i}$ is determined by the number of applications susceptible to that exposure \cite{app-android}. $d_i(\mathbf{a})$ is the cost of deployment of legitimate/malicious behaviour by application $i$ given the permissions sought by all applications for legitimate/malicious functionality. Further, we assume $\forall i \in [n], \tilde c_{\tau,i} \geq 100 \times c_{\tau,i}$, to ensure that the malicious action payoff is greater than the legitimate action payoff in general. The Intention $\xi_i$ of a mobile application $i$ reflects its choice between legitimate and malicious functionality. The members of $\mathcal{G}_{app} := ([n], \{ A_i \}_{i \in [n]}, \{ (u_i,v_i) \}_{i \in [n]})$ are given below:
\vspace*{-3pt}
\begin{enumerate}
\item $A_i := \{ \textsc{userexp}, \textsc{btheft}, \textsc{ahijack}, \textsc{shijack}, \textsc{binject}, \textsc{b-nocheck}, \textsc{alaunch}, \textsc{slaunch} \}$
\item $\forall i \in [n], \forall \mathbf{a} \in \mathbf{A}$:
\[
  u_i(\mathbf{a}) =
  \begin{cases}
                                   0 & \text{if $a_i \neq \textsc{userexp}$} \\
                                   c_{\tau,i} - d_i(\mathbf{a}) & \text{if $a_i = \textsc{userexp}$} \\
  \end{cases}
\]
\item $\forall i \in [n], \forall \mathbf{a} \in \mathbf{A}$:
\[
  v_i(\mathbf{a}) =
  \begin{cases}
  				 c_{\tau,i} - d_i(\mathbf{a}) & \text{if $a_i = \textsc{userexp}$} \\
                                   0.56 \times \tilde c_{\tau,i} - d_i(\mathbf{a}) & \text{if $a_i = \textsc{btheft}$} \\
                                   0.03 \times \tilde c_{\tau,i} - d_i(\mathbf{a}) & \text{if $a_i = \textsc{ahijack}$} \\
                                   0.81 \times \tilde c_{\tau,i} - d_i(\mathbf{a}) & \text{if $a_i = \textsc{shijack}$} \\
                                   0.44 \times \tilde c_{\tau,i} - d_i(\mathbf{a}) & \text{if $a_i = \textsc{binject}$} \\
                                   0.87 \times \tilde c_{\tau,i} - d_i(\mathbf{a}) & \text{if $a_i = \textsc{b-nocheck}$} \\
                                   0.43 \times \tilde c_{\tau,i} - d_i(\mathbf{a}) & \text{if $a_i = \textsc{alaunch}$} \\
				 0.86 \times \tilde c_{\tau,i} - d_i(\mathbf{a}) & \text{if $a_i = \textsc{slaunch}$} \\
  \end{cases}
\]
\end{enumerate}

\noindent It is easy to see that $\forall i \in [n], \forall \mathbf{a} \in \mathbf{A}, v_i(\mathbf{a}) \geq u_i(\mathbf{a})$, and $v_i(\cdot)$ can change in each iteration of the Intention Game depending on the operating system state. Further, given a game evolution, if too many attacks are seen as per strategy of some player(s), players can conclude that the players corresponding to the illegitimate actions are playing best responses under their self reflection of the Intention Game, and can employ the participational equilibria (Section \ref{subsec:eqm}) to determine parameters of coexistence: in this instance, deciding whether a set of competing applications should be simultaneously serving the user given that some of them are malicious. \\

\subsection{A Intention Game between ISPs sourcing data from CDNs}
\label{subsec:cdnisp-ig}

\noindent Internet customers avail the Internet as a service through Internet Service Providers (ISPs), with a prominent class of traffic being HTTP objects \cite{networks}. In turn, the ISPs source HTTP objects from public Content Delivery Networks (CDNs) as a commodity, and in the event of a network congestion, can have spontaneous contracts with CDNs for end-user to server assignment and in-network server allocation to reduce the HTTP object download time \cite{cdn-isp-coll} (which ensure a certain quality of service even in cases of a network congestion), for overall better payoffs by serving their customers even in a situation of low availability. This results in a competitive game where different ISPs can compete for data consumption from a common set of CDNs, given a possibly uncertain underlying network state. \\ 

\noindent We state this game formally. We have a set of ISPs $[n]$ in need of HTTP objects of quantity $q$, per CDN. We are given a public CDN community (set) \textbf{C} and possibility hidden reactive contracts of each ISP with the CDN community \textbf{H} (the choice of which gives ISP $i$'s Intention $\xi_i$). We have a network $\mathcal{N}$ interlinking the CDNs and ISPs, which is prone to bottlenecks. The consequent CDN-ISP $k$-Intention Game is $\mathcal{G}_{isp} := ([n], \{ A_i \}_{i \in [n]}, \{ (u_i,v_i) \}_{i \in [n]})$, where: \\
\vspace*{-17pt}
\begin{enumerate}
\item $A_i := \{ a_i = (q_1, q_2, ..., q_{|\textbf{C}|})^T : \forall j \in [|\textbf{C}|], q_j \in \mathbb{Z}_+ \}$, \hspace*{3pt} $\forall i \in [n]$.
\item $\forall i \in [n], \forall \mathbf{a} \in \mathbf{A}, u_i(\mathbf{a}) := a_i^TP_i(a_i) - C_\mathcal{N}(\mathbf{a})$.
\item $H_i:$ If $\mathcal{N}$ is under a severe congestion from ISP $i$'s perspective, allow $i$ to source the same amount of HTTP objects from the CDNs at fraction $\rho_\tau \in (0,1)$ of the cost via end-user to server assignment and in-network server allocation, under a contract $\mathbf{H}$ adopted in reaction to the network congestion, to bring down the average cost per object. \\
$\forall i \in [n], \forall \mathbf{a} \in \mathbf{A}, v_i(\mathbf{a}) := a_i^TP_i(a_i) - C_{\mathbf{H}}(\mathbf{a})$ under contract $\mathbf{H}$, with $C_{\mathbf{H}}(\mathbf{a}) = \rho_\tau \times C_\mathcal{N}(\mathbf{a})$.
\end{enumerate}

\noindent We state that this is a skeletal construction, with uninstantiated price $P_i$ which depends on each ISP $i$'s serving it's customers alone and cost $C_\mathcal{N}$ which depends on the collective traffic of all ISPs dealing with $\textbf{C}$, as a function of the network state/availability. Further, the cost fraction $\rho_\tau$ can be a function of the iteration of the Intention Game, say for example $\rho_\tau := (\frac{3}{4})^\tau$, implying the private payoff can change per iteration of the Intention Game. This highlights the disparity between the private and public payoffs, under a dynamic network state, leading to different average cost for HTTP objects. It is this spontaneity in overall private payoff determination that is intrinsic to an Intention Game.

\subsection{A Bitcoin Mining Intention Game}
\label{subsec:btc-game}

\noindent Bitcoin is a blockchain protocol that uses a solution to a cryptographic hash puzzle to achieve block consensus among a set of nodes/players in a network \cite{mining-game}. The process of finding a solution to the puzzle is called mining, and the probability of success of any miner/player is proportionate to its compute (hash) power. Given the possibility of forks in the underlying blockchain, legitimate mining strategy dictates to continue mining on the longest existing branch in a fork, at any time step. However, there is a possibility of some subset of miners to collude and chose a shorter branch in a fork and invest their collective compute power there for mining rewards. This illegitimate strategy is called selfish mining. \\
\noindent We define a $k$-Intention Game $\mathcal{G}_{btc}$ between miners where there is a possibility of $k$ miners defecting towards selfish mining for better reward as opposed to legitimate mining (called \textsc{frontier} strategy \cite{mining-game}). We will consider the Intention Game is in iteration $\tau$, given the last confirmed block as $B^\tau$. We will denote the hash power of player $i$ by $c_i$, and approximate the hash power of upto $k \hspace*{3pt} (> 1)$ selfish mining colluders $p^+ \subseteq [n], |p^+| \le k$ after $B^\tau$ is confirmed as $\tilde c_{\tau,p^+} \hspace*{3pt} ( \geq c_i, \forall i \in [n])$. We will denote the set of forks originating from $B^\tau$ as $\Lambda^\tau$. We will denote no mining by $\textsc{none}$, and honest and dishonest mining strategies by $\textsc{frontier}$ and $\textsc{selfish}$ respectively. The Intention $\xi_i$ of miner $i$ is a choice between legitimate and selfish mining. We detail the game $\mathcal{G}_{btc} := ([n], \{A_i\}_{i \in [n]}, \{(u_i,v_i)\}_{i \in [n]})$ next:
\vspace*{-3pt}
\begin{enumerate}
\item $A_i = \{ (B^\tau, \lambda, \textsc{mine-type}): \tau \in \mathbb{Z}_+, \lambda \in \Lambda^\tau, \textsc{mine-type} \in \{ \textsc{none}, \textsc{frontier}, \textsc{selfish} \} \}$
\item $\forall i \in [n], \forall \mathbf{a} \in \mathbf{A}$:
\[
  u_i(\mathbf{a}) =
  \begin{cases}
                                   0 & \text{if $a_i = (B^\tau, \lambda, \textsc{mine-type})$, $\textsc{mine-type} \neq \textsc{frontier}$} \\
                                   c_i & \text{if $a_i = (B^\tau, \lambda, \textsc{frontier})$, $\lambda$ is the longest fork} \\
  \end{cases}
\]
\item $\forall i \in [n], \forall \mathbf{a} \in \mathbf{A}$:
\[
  v_i(\mathbf{a}) =
  \begin{cases}
                                   0 & \text{if $a_i = (B^\tau, \lambda, \textsc{mine-type})$, $\textsc{mine-type} = \textsc{none}$} \\
                                   c_i & \text{if $a_i = (B^\tau, \lambda, \textsc{frontier})$, $\lambda$ is the longest fork} \\
                                   \tilde c_{\tau,p^+} & \text{if $\exists p^+ \subseteq [n], |p^+| \le k, \forall j \in p^+ \hspace*{5pt} a_j = (B^\tau, \lambda, \textsc{selfish})$,} \\
                                   & \text{$i \in p^+$ and $\lambda$ is not the longest fork.} \\
  \end{cases}
\]
\end{enumerate}

\noindent As is seen in the definition of $\mathcal{G}_{btc}$, the private payoff on selfish mining by a mining pool $p^+$ of players is higher than what each player can achieve individually through honest mining. Further, this private payoff $\tilde c_{\tau,p^+}$ can change per iteration of the game, depending on what transactions relevant to the selfish miners $p^+$ go on-chain after $B^\tau$, as the miners can be selective/preferential on the transactions they put (for faster confirmation) in the blocks on the fork $\lambda$ chosen for selfish mining.

\subsection{An Intention Game Use to discover Participants in a Secret Sharing Protocol}
\label{subsec:key-game}
We give a $1$-Intention Game $\mathcal{G}_{ss}$ where given $n$ players, there are $m (\leq n)$-honest but unknown players that intend to share a secret, as part of a secret sharing protocol \cite{pvss}. In the game, the dealer \textbf{D} visits the players in some order and give a secret fragment. Honest players accept the fragment and defect from the public image of the Intention Game under a secret payoff from the dealer \textbf{D}, but dishonest players reject the fragment and conform to the public image of the Intention Game. Although it may appear counterintuitive that an honest player accepts a hidden payoff on a visit from a dealer (which reflects its Intention $\xi$), it should be noted that this Intention Game is only defined to achieve a cryptographic goal of secret sharing among an uncertain (in terms of participation) pool of players. \\

\noindent We formally state the Intention Game for Participant Discovery in a Secret Sharing Protocol $\mathcal{G}_{ss} := ([n], \{A_i\}_{i \in [n]}, \{(u_i,v_i)\}_{i \in [n]})$. We assume there is a hidden dealer \textbf{D} that visits all players. The dealer chooses an arbitrary permutation of $[n]$ to visit one player in each iteration of $\mathcal{G}_{ss}$. Each player has three possible actions: either it is unvisited (denoted by \textsc{wait}), or it is visited and rejects the dealer's proposed secret fragment (denoted by \textsc{recv-rej}), or  it is visited and accepts the dealer's proposed secret fragment (denoted by \textsc{recv-acc}). The strategies, contract and payoffs are formally given by: \\
\vspace*{-3pt}
\begin{enumerate}
\item $A_i := \{ \textsc{wait}, \textsc{recv-rej}, \textsc{recv-acc} \}$
\item $\forall i \in [n], \forall \mathbf{a} \in \mathbf{A}$:
\[
  u_i(\mathbf{a}) =
  \begin{cases}
                                   0 & \text{if $a_i = \textsc{recv-acc}$} \\
                                   1 & \text{if $a_i \in \{ \textsc{wait}, \textsc{recv-rej} \}$} \\
  \end{cases}
\]
\item $\forall i \in [n], \forall \mathbf{a} \in \mathbf{A}$:
\[
  v_i(\mathbf{a}) =
  \begin{cases}
                                   1 & \text{if $a_i \in \{ \textsc{wait}, \textsc{recv-rej} \}$} \\
                                   2 & \text{if $a_i = \textsc{recv-acc}$} \\
  \end{cases}
\]
when dealer $\mathbf{D}$ visits player $i$.
\end{enumerate}

\noindent Finally, for a requirement that $m (\leq n)$ participant players are discovered for the secret sharing, the Intention Game $\mathcal{G}_{ss}$ can run for the smallest $\tau$ till a $(\tau, m)$-Honesty Equilibrium is achieved. In the worst case, $\tau = n$.

%% file: 05_dual-related.tex
\section{Selfless Games: Cases of Private Payoff being less than the Public Payoff}
\label{sec:selfless}

Now, we give a dual framework $\tilde{\mathcal{G}}$ of an Intention Game $\mathcal{G}$. A `Selfless Game' $\tilde{\mathcal{G}}$ is a dual of $\mathcal{G}$ in the sense that for all iterations of $\tilde{\mathcal{G}}$, the private payoff is less than the public payoff: $\forall i \in [n], \forall \mathbf{a} \in \mathbf{A}, v_i(\mathbf{a}) \leq u_i(\mathbf{a})$, but for all iterations, there exists a private (to $i$) increment $h_i(\mathbf{a})$ \emph{promised} to be revealed in a future iteration, such that $\forall i \in [n], \forall \mathbf{a} \in \mathbf{A}, v_i(\mathbf{a}) + h_i(\mathbf{a}) \geq u_i(\mathbf{a})$. \\

\noindent Note that this definition of a game where best responses don't maximize payoffs would seem counterintuitive but still is relevant: such incentivizations are possible and rooted in behavioural psychology. For instance, consider a game between a mother and a son in which best responses coming from the mother are suboptimal under  $\mathbf{Im}(\tilde{\mathcal{G}})$. However, the `action profile certificates' might incentivize the mother to get an alternative reward from the father in future, compensating the sub optimality achieved in the present iteration of $\tilde{\mathcal{G}}$. Thus although the mother does not win against the son, she ultimately wins due to a `promise' with the father. \\

\noindent Having suggested a possible dual of Intention Games, we conclude that a rigorous motivation and definition of Selfless Games is not in the scope of our current work. \\

\section{Related Work}
\label{sec:related}

Classical game theory considers extensively models of competition in the form of games with incomplete information \cite{gt}. Bayesian games (Chapter 9, \cite{gt}) allow players to have incomplete information about the payoffs of other players, but have beliefs about those payoffs through some probability distribution \cite{bgi}. Bayesian rationality dictates that players choose to maximize their expected utility \cite{epis-book}. The said Bayesian rationality also governs decision making in epistemic game theory \cite{epis-stanford}, but rational decisions by a player in this framework are evaluated on the basis of what a player's preferences are, and beliefs on what their competitors are going to do \cite{epis-stanford}. These beliefs can be characterized using Harsanyi's type spaces, or probabilistic type spaces, among others \cite{epis-econ}. Stochastic games \cite{stochastic} are extensive form games where the game transitions taken by players are random variables: the next game is a probabilistic function of the current game and current joint action / strategy profile. In all these cases, the game theoretic model is of incomplete information, but can be modelled as a distribution on the underlying game, or some distribution on the uncertainty of the choices of the players' opponents. However, Intention Games do not permit any freedom to model / analyze structure of the private payoffs, which can change arbitrarily per iteration of the game, and consequently, for any player, there is a uniform distribution on the payoffs chosen by other players in the worst case, rendering the player blind on the underlying true game. \\

\noindent There do exist studies on learning in games \cite{l-gt}. Again, methods such as fictitious play are inapplicable since the private payoff can change \emph{per iteration} of the game. Also no player can learn the exact private payoff of any other player, only the lower bound on the payoff defection of each player through the public information given by all players during the game. Principles from contract theory \cite{contract} are rendered inapplicable due to the spontaneity and temporality of Intention Game agreements. Furthermore, the information asymmetry \cite{info-asymm} model in contract theory assumes at least one party has complete information, which is untrue for $k (> 1)$ Intention Games (Definition \ref{def:k-ig}). Finally, as opposed to mixed oligopolies \cite{mixed-olig}, where players/firms are hidden, Intention Games permit true payoffs to be hidden with all players public. \\

\noindent Traditional game theory considers rational choices and utility maximization as a norm. However, there have been models for defection from rational behaviour owing to beliefs, social issues, group issues under the formalism of behavioural game theory \cite{bgt}. Even so, the defections considered in behavioural game theory occur from eccentricities implicit in the participating players due to external effects. In comparison, we do not compromise on the notion of rationality in defining Intention Games. While considering players to be rational, but partially honest and blind, we propose our competitive framework. \\

\noindent In Bayesian Games with Intentions \cite{bgi}, Bayesian Games are extended by defining Intentions as a mapping from types to strategies, to address/represent the complete space of preferences as a consequence of beliefs, which are not addressable through traditional Bayesian games. There have been works \cite{i-and-gt}, to draw connections between the informal philosophical connotation of intentions, and the formal mathematical construction of game theory. There have been qualitative studies (without probabilistic analysis and utilities) of modelling rationality and intentions in a rational choice framework \cite{decision}, which includes multi-agent strategic interactions as in decision theory and game theory. Certain contributions have defined a player's intentions as actions taken by the player as function of the other players' actions, given a repeated game \cite{intent-recog}. They conclude that various equilibria are achieved (including Nash and Stackelberg) when all players' intentions are optimised. However, none of the existing `intentions' within game theory define an a new paradigm of strategic interaction given high uncertainty in the knowledge of the true game, as is proposed by our Intention Games framework. \\

%% file: 06_conclusion.tex
\section{Conclusions and Future Work}
\label{sec:conclude}
In this contribution, we have proposed a new line of enquiry for partially honest and blind interactions (in the knowledge of the true game) among competitive players. Our novel game theoretic construction called Intention Games is an extension of conventional strategic form games, and is more general than Bayesian games, allowing simultaneous dishonesty among players and consequently permitting high uncertainty in the underlying true game. We have defined new notions of equilibria in our novel game setting, to determine parameters for participation and coexistence among potentially dishonest players. We have given an example of an Intention Cournot Duopoly to demonstrate the implications of publicly devious behaviour on payoff outcomes. We have given three constructions of computational Intention Games between mobile applications, Internet service providers, and Bitcoin miners, respectively. We have also demonstrated how Intention Games can be used for the task of distributing shares of a cryptographic key, among a pool of uncertain protocol participants. Finally, we have outlined a dual framework of Intention Games, which we call Selfless Games. \\

\noindent In future, we would like to explore participational equilibria for $k$-Intention Games with $k > 1$, and define an equilibrium concept to measure defection centrality.  We would like to do a rigorous formal theoretical investigation of the computation complexity of the existing (and any consequent) definitions of participation equilibria. We would like to formally define social welfare as a function of the participational equilibria. We would also like to explore a rigorous motivation and definition of Selfless Games. Finally, we would like to build cryptographic primitives with underlying hardness coming from an Intention Game framework. \\

%% file: 07a_toymodel.tex
\newpage
\section{An Example of a Two Player Intention Game}
\label{app:toy}

We give a 2-player example of a general Intention Game wherein we demonstrate, given the choices of players to conform to or defect from the public image of the Intention Game, what is the game perceived by each player, and which game is actually played. \\

\begin{figure}
\label{fig:toygame}
\centering
\includegraphics[width=1.0\textwidth, height=0.48\textheight]{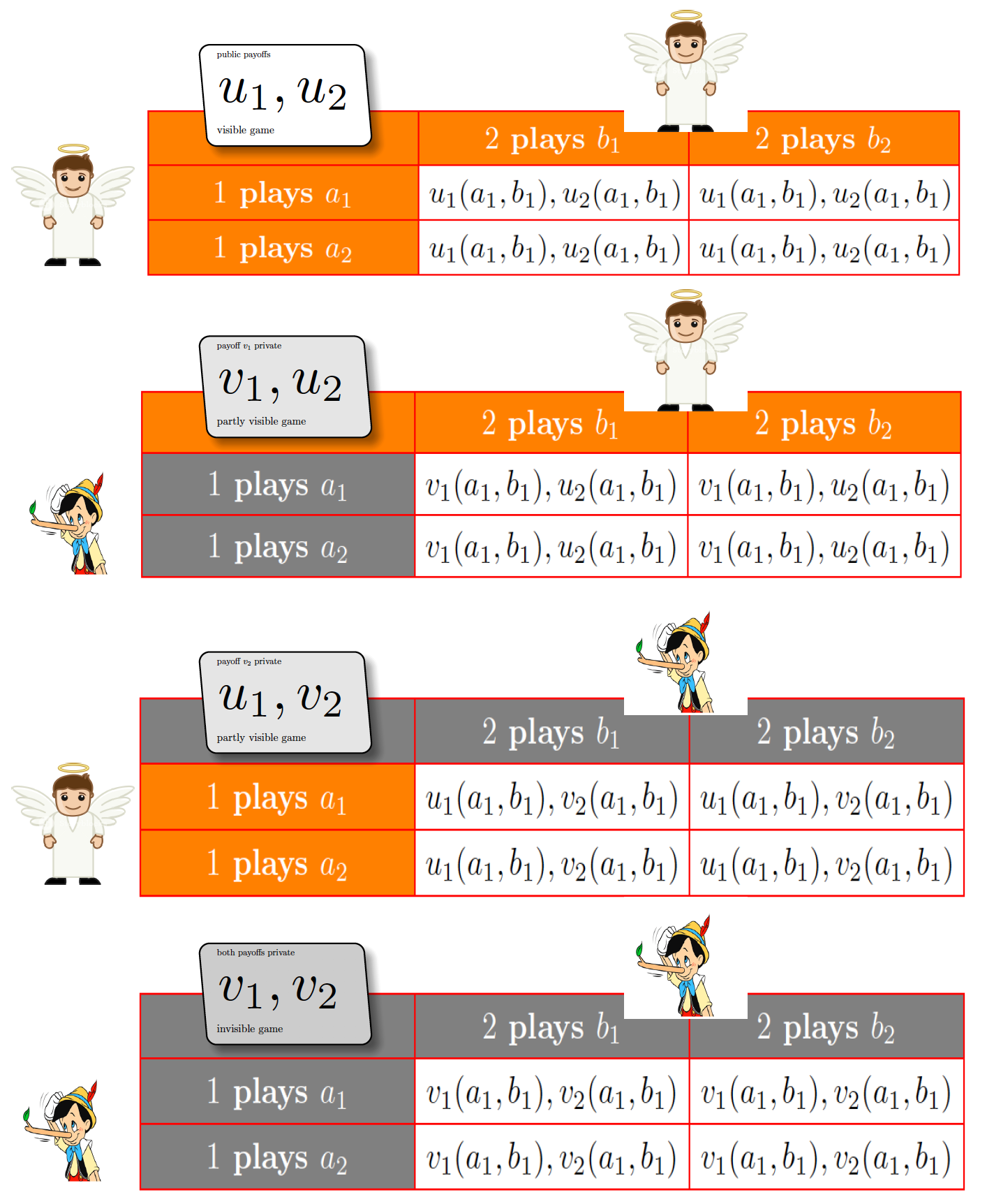}
\caption{A toy Intention Game with two players and two actions per player. A single Intention Game can have four evolutions per iteration.}
\end{figure}

\noindent Referring to Figure $2$, consider the 2-player Intention Game $\mathcal{G} = ([2], \{ A_1 = \{ a_1,a_2 \},  B_2 = \{ b_1,b_2 \} \}, \{(u_i,v_i)\}_{i \in [2]}, \{ \xi_i \}_{i \in [2]})$. In the first case, both players are honest, and agree on the public image of the Intention Game $\mathbf{Im}(\mathcal{G}) = ([2], \{ A_1 = \{ a_1,a_2 \},  B_2 = \{ b_1,b_2 \} \}, \{u_i\}_{i \in [2]})$ (denoted by the two angels), and consequently the perceived game by both players and the true game is $\mathbf{Im}(\mathcal{G})$. In the second case, player-1 defects (denoted by the Pinnochio) and player-2 is honest (denoted by the angel). The true game and the game perceived by player-1 is $\mathbf{Ref}_1(\mathcal{G}) = ([2], \{ A_1 = \{ a_1,a_2 \},  B_2 = \{ b_1,b_2 \} \}, \{ v_1, u_2 \})$, but player-2 believes the game is $\mathbf{Im}(\mathcal{G})$. The third case is the exact opposite of the second case, where the true game and the game perceived by player 2 is $\mathbf{Ref}_2(\mathcal{G}) = ([2], \{ A_1 = \{ a_1,a_2 \},  B_2 = \{ b_1,b_2 \} \}, \{ u_1, v_2 \})$ and the game perceived by player 1 is $\mathbf{Im}(\mathcal{G})$. The final case is the one where neither of the two players conceives the true game: player 1 believes the game is $\mathbf{Ref}_1(\mathcal{G})$, player 2 believes the game is $\mathbf{Ref}_2(\mathcal{G})$, but the true game is $([2], \{ A_1 = \{ a_1,a_2 \},  B_2 = \{ b_1,b_2 \} \}, \{v_i\}_{i \in [2]})$.